\documentclass[11pt,a4paper]{amsart}
\usepackage{amsmath,amssymb,amsthm,amscd,verbatim,enumerate}
\usepackage{graphicx,subfigure}
\usepackage[utf8]{inputenc}
\usepackage[colorlinks=true,citecolor=black,linkcolor=black,urlcolor=blue]{hyperref}
\usepackage[lmargin=31mm,rmargin=31mm,bmargin=31mm,tmargin=31mm]{geometry}
\usepackage{multirow}

\allowdisplaybreaks

\theoremstyle{plain}
\newtheorem{theorem}{Theorem}[section]
\newtheorem{lemma}[theorem]{Lemma}

\theoremstyle{definition}

\newcommand{\set}[1]{\{#1\}}
\let\leq\leqslant
\let\geq\geqslant
\let\subset\subseteq

\title{Making Triangles Colorful}

\author[Cardinal]{Jean Cardinal}
\author[Knauer]{Kolja Knauer}
\author[Micek]{Piotr Micek}
\author[Ueckerdt]{Torsten Ueckerdt}
\thanks{Research supported by the ESF EUROCORES programme EuroGIGA, CRP ComPoSe ({\tt http://www.eurogiga-compose.eu}), and GraDR ({\tt http://kam.mff.cuni.cz/gradr/}). Jean Cardinal is supported by the F.R.S-FNRS as grant convention no.\ R 70.01.11F within the ComPoSe project. Kolja Knauer is supported by DFG grant FE-340/8-1 as part of the project GraDR. Piotr Micek is supported by the Ministry of Science and Higher Education of Poland as grant no.\ 884/N-ESF-EuroGIGA/10/2011/0 within the GraDR project.}

\address{Universit\'e Libre de Bruxelles\\
Computer Science Department\\
Brussels\\
Belgium
}

\address{Technische Universit\"at Berlin\\
Institut f\"ur Mathematik\\
Berlin\\ 
Germany
}

\address{Theoretical Computer Science Department\\
         Faculty of Mathematics and Computer Science\\
         Jagiellonian University\\
         Krak\'{o}w\\
         Poland}

\address{Department of Mathematics\\
	  Karlsruhe Institute of Technology\\
        Karlsruhe\\
        Germany
}

\email{jcardin@ulb.ac.be}
\email{knauer@math.tu-berlin.de}
\email{Piotr.Micek@tcs.uj.edu.pl}
\email{Torsten.Ueckerdt@kit.edu}

\begin{document}
\sloppy 
\begin{abstract}
We prove that for any point set $P$ in the plane, a triangle $T$, and a positive integer $k$, there exists a
coloring of $P$ with $k$ colors such that any homothetic copy of $T$ containing at least $ck^8$ points of $P$, for some constant $c$, contains at least one of each color. This is the first polynomial bound for range spaces induced by homothetic polygons.
The only previously known bound for this problem applies to the more general case of octants in $\mathbb{R}^3$, but is doubly exponential. 
\end{abstract}

\maketitle

\section{Introduction}


Covering and packing problems are ubiquitous in discrete geometry. In this context, the notion of $\epsilon$-nets
captures the idea of finding a small but representative sample of a data set (see for instance Chapter 10 in Matou\v{s}ek's lectures~\cite{Mat}). Given a set system, or {\em range space}, on $n$ elements, an $\epsilon$-net for this system is a subset of the elements such that any range containing at least a fraction $\epsilon$ of the whole set contains at least one element of the subset.

In this paper, we are interested in {\em coloring} a point set so that any range containing sufficiently many points contains at least
{\em one point of each color}. Hence instead of finding a single subset of representative elements, we wish to partition the point set
into representative classes. 

We are interested in range spaces defined by a (countable or finite) point set $P\subset\mathbb{R}^d$ and a collection $\mathcal B$ of subsets of $\mathbb{R}^d$. In what follows, we are mainly concerned with the case where $\mathcal B$ is a collection of {\em convex bodies}, that is, compact convex subsets of $\mathbb{R}^d$. The corresponding set system is the collection of subsets of $P$, called {\em ranges}, formed by intersecting $P$ with a member of $\mathcal B$. This construction yields so-called {\em primal} range spaces. For instance, if $P$ is a set of points in the plane and $\mathcal B$ the set of all disks, then the ranges are all possible intersections of $P$ with a disk.

One can also consider {\em dual} range spaces, where the ground set is a (countable or finite) collection $\mathcal B$ of subsets of $\mathbb{R}^d$, and the ranges are all subsets of $\mathcal B$ that have a point of $\mathbb{R}^d$ in common. In general, those are also referred to as (primal and dual) {\em geometric hypergraphs}.

For a given family of range spaces, we define the function $p(k)$ as the minimum number $p$ such that the following holds: every point set can be colored with $k$ colors such that any range containing at least $p$ points contains at least one of each color.

It is not difficult to show that if $p(k)=O(k)$ for a family of range spaces, then this family admits $\epsilon$-nets of size $O(1/\epsilon )$ (see the discussion in Pach and Tardos~\cite{PT10c}). In the case of dual range spaces, the problem of bounding $p(k)$ is known as {\em covering decomposition}. In this setting, we are given a collection of objects, and we wish to partition them into $k$ classes, so that whenever a point is contained in sufficiently many objects of the initial collection, it is contained in at least one object of each class.

We prove a polynomial upper bound on $p(k)$ for primal range spaces induced by homothetic triangles in the plane.

\subsection{Previous works.}

These questions were first studied by J{\'a}nos Pach in the early eighties~\cite{Pa80}. An account of early related results and conjectures can be found in Chapter 2 of the survey on open problems in discrete geometry by Brass, Moser, and Pach~\cite{BMP}. 

In the past five years, tremendous progress has been made in this area, for various families of range spaces. One of the most striking 
achievements is the recent proof that $p(k)=O(k)$ for translates of convex polygons, the culmination of a series of intermediate results for various special cases. We refer the reader to Table~\ref{tab:summary} for a summary of the known bounds. 

\begin{table}[!ht]
\footnotesize\center
\begin{tabular}{|c||c|c||c|c|}
\hline

Range spaces      & primal &  dual \\

\hline 
\hline

\multirow{2}{*}{halfplanes} &  \multirow{2}{*}{$p(k)=2k-1$~\cite{ACCLS09,K11,SY10}} &  ${p}(2)=3$~\cite{F10} \\
           &           & ${p}(k)\leq 4k-3$~\cite{ACCLS09,SY10}  \\
 
\hline
\hline

translates of     &  \multicolumn{2}{c||}{\multirow{2}{*}{$p(k)=O(k)$~\cite{TT07,PT10,PT09,ACCLOR10,GV11}}} \\
a convex polygon  &  \multicolumn{2}{c||}{} \\

\hline 

translates of &  \multicolumn{2}{c||}{$p(2)\leq 12$~\cite{KP11}} \\
octant in $\mathbb{R}^3$  & \multicolumn{2}{c||}{$p(k)\leq 12^{2^k}$~\cite{KP12}} \\

\hline 

unit disks & \multicolumn{2}{c||}{$p(2)\leq 33$\footnotemark} \\

\hline
\hline

\multirow{2}{*}{bottomless rectangles} &  $p(2)=4$~\cite{K11} &  ${p}(2)=3$~\cite{K11} \\
                      &  $1.6 k \leq p(k)\leq 3k-2$~\cite{A12} & ${p}(k)\leq 12^{2^k}$~\cite{KP12} (from octants in $\mathbb{R}^3$) \\ 

\hline

axis-aligned & \multirow{2}{*}{$\infty$~\cite{CPST09}} &  \multirow{2}{*}{$\infty$~\cite{PTT05}} \\
  rectangles &  & \\

\hline
\hline

disks and &  \multirow{2}{*}{$\infty$~\cite{PTT05}} &  \multirow{2}{*}{Open} \\
halfspaces in $\mathbb{R}^3$                            &                                         &  \\

\hline
\end{tabular}
\vspace{2ex}
\caption{\label{tab:summary}Known results for various families of range spaces.
For range spaces induced by translates of a set, the primal problem is the same as the dual.
When more than one reference is given, they correspond to successive improvements, but only the best known bound is indicated.  
The symbol $\infty$ indicates that $p(k)$ does not exist. For dual range spaces induced by disks, it is not known
whether $p(k)$ exists.}
\end{table}
\footnotetext{This result has apparently been proved by Mani-Levitska and Pach in 1986, although it seems that the
corresponding unpublished manuscript is now lost.}

The specific case of translates of a triangle with $k=2$ was tackled by Tardos and T{\'oth} in 2007~\cite{TT07}. They proved that every point set can be colored red and blue so that every translate of a given triangle containing at least 43 points contains at least one red and one blue. We generalize this result in two ways: we consider {\em homothetic} triangles, and an arbitrary number of colors.

The only previously known results applying to our problem are due to Keszegh and P{\'a}lv{\"o}lgyi~\cite{KP11,KP12}. They actually apply to the more general case of translates of (say) the positive octant in a cartesian representation of $\mathbb{R}^3$. 
The special case of triangles homothetic to the triangle with vertices $(0,0)$, $(1,0)$ and $(0,1)$ occurs when all points lie on a plane orthogonal to the vector $(1,1,1)$. The bound that was proven for arbitrary $k$ is of the order of $12^{2^k}$, and is most probably far from tight.

\subsection{Our Result.}

\begin{theorem}
\label{th:main}
Given a finite point set $P\subset \mathbb{R}^2$, a triangle $T\subset\mathbb{R}^2$ and a positive integer $k$, there exists a
coloring of $P$ with $k$ colors such that any homothetic copy of $T$ containing at least $144\cdot k^8$ points of $P$ contains
at least one of each color.
\end{theorem}

The proof is elementary, and builds on the previous work by Keszegh and P{\'a}lv{\"o}lgyi~\cite{KP11,KP12}. 
The degree of the polynomial depends on $p(2)$. Hence any improvement on $p(2)$
would yield a polynomial improvement in the bound. For the same reason, it can be shown that the same coloring method cannot 
be used to prove any upper bound better than $O(k^4)$ (as $p(2)\geq 4$).

\section{Proof}

\begin{lemma}
If $p(2)\leq c$, for some constant $c$, then $p(2k)\leq c^2 p(k)$, for all $k\geq2$.
\end{lemma}
\begin{proof}
It suffices to prove the lemma for any fixed triangle $T$ and then argue for all others using an affine transformation of the plane. 
Let $T$ be the triangle with vertices $(0,0)$, $(1,0)$ and $(0,1)$. 

Consider a $k$-coloring $\phi$ of $P$ such that any homothetic copy of $T$ containing at least $p(k)$ points contains one of each color. 
Label the colors of $\phi$ by $1,\ldots,k$. 

We now describe a simple procedure to double the number of colors. 
For $1\leq i\leq k$ let $P_i=\phi^{-1}(i)$ that is the set of points with color $i$. 
Provided $p(2)\leq c$ there is a 2-coloring $\phi_i:P_i\to\set{i',i''}$ of $P_i$ such that for any homothetic copy $T'$ of $T$ containing at least $c$ points of $P_i$, $T'$ contains at least one point of each color. 
We claim that $\phi'=\bigsqcup\phi_i$ is a $2k$-coloring of $P$ such that for any homothetic copy $T'$ of $T$ containing at least $c^2p(k)$ points, $T'$ contains at least one point of each of the $2k$ colors. 

Consider a homothetic copy $T'$ of a triangle $T$ containing at least $c^2p(k)$ points from $P$, and in order to get a contradiction suppose that one of the $2k$ colors used by $\phi'$ is missing in $T'$. 
Let $i'$ be this color. Note that if there are at least $c$ points in $T'$ with color $i$ then $i'$ and $i''$ must be present in $T'$, from the correctness of the 2-coloring $\phi_i$. 
Hence we conclude that there are less than $c$ points in $T'$ with color $i$.

Order the points in $T'\cap P=\{p_1,p_2,\ldots \}$ in such a way that the sum of their $x$- and $y$-coordinates is non-decreasing. Hence the order corresponds to a sweep of the points in $T'\cap P$
by a line of slope $-1$. By the pigeonhole principle, since there are less than $c$ points colored with $i$, there must exist a subsequence $Q=(p_j, p_{j+1},\ldots ,p_{j+\ell-1})$ of points of color distinct from $i$, 
of length $\ell := c^2p(k)/c=cp(k)$.

Let $R := P_i\cap \{p_1, p_2,\ldots ,p_{j-1}\}$ be the set of points of color $i$ that come before $Q$ in the sweep order. 
By assumption, we have $|R|<c$. Hence the points of $Q$ can be covered with $c$ translates of the first quadrant, such that none of them intersects $R$. Applying the pigeonhole principle a second time, one of these quadrants must contain at least $|Q|/c = cp(k)/c = p(k)$ points, none of which is colored $i$.
This quadrant, together with the sweepline containing the last point $p_{j+\ell-1}$ of $Q$, forms a triangle that is homothetic to $T$, contains at least $p(k)$
points, none of which has color $i$. This is a contradiction with the correctness of the initial $k$-coloring $\phi$.
\end{proof}

\begin{proof}[Proof of Theorem~\ref{th:main}]
It was shown by Keszegh and P{\'a}lv{\"o}lgyi that $p(2)\leq 12$~\cite{KP11}.
Hence it remains to solve the recurrence of the previous lemma with $c=12$.
We look for an upper bound on $p(k)$ satisfying $p(2k) \leq 144\cdot p(k)$ for any positive integer $k$, and $p(2)\leq 12$.
This yields $p(2^i)\leq 144^i$ for any positive integer $i$, and
$p(k)\leq 144^{\lceil \log_2 k\rceil}<144\cdot k^8$ for any positive integer $k$. 
\end{proof}

\section{Open problems}

The only lower bound on $p(k)$ the authors are aware of is the construction for bottomless rectangles in \cite{A12}, namely $p(k)\geq 1.6k$. 

The dual problem for triangles remains a challenge: what is the minimum number ${p}(k)$ such that for any triangle $T$, any family $\mathcal{T}$ of homothetic copies of $T$ can be $k$-colored in such a way that any point 
covered by at least ${p}(k)$ triangles from $\mathcal{T}$ is covered by at least one triangle of each color? 
Up to now, the only upper bound for this problem is given by the upper bound for octants in $\mathbb{R}^3$ and is doubly exponential in $k$.

Finally, no bound at all is known for the primal range space induced by axis-aligned squares: does there exist a function $p(k)$ such that for any point set $P$ there is a $k$-coloring of $P$ such that any axis-aligned square containing at least $p(k)$ points of $P$ contains at least one point of each color?

\section{Acknowledgments}

This work was carried out during a visit of Kolja, Piotr, and Torsten in Brussels, funded by a EUROCORES Short Term Visit grant, in the framework of the EuroGIGA programme.  
The authors wish to thank Stefan Langerman for insightful discussions. 

\bibliographystyle{plain}
\bibliography{biblio_p(k)_vs_c(k)}

\begin{thebibliography}{10}

\bibitem{ACCLOR10}
Greg Aloupis, Jean Cardinal, S{\'e}bastien Collette, Stefan Langerman, David
  Orden, and Pedro Ramos.
\newblock Decomposition of multiple coverings into more parts.
\newblock {\em Discrete {\&} Computational Geometry}, 44(3):706--723, 2010.

\bibitem{ACCLS09}
Greg Aloupis, Jean Cardinal, S{\'e}bastien Collette, Stefan Langerman, and
  Shakhar Smorodinsky.
\newblock Coloring geometric range spaces.
\newblock {\em Discrete {\&} Computational Geometry}, 41(2):348--362, 2009.

\bibitem{A12}
Andrei Asinowski, Jean Cardinal, Nathann Cohen, S\'ebastien Collette, Thomas
  Hackl, Michael Hoffmann, Kolja Knauer, Stefan Langerman, Micha{\l} Laso\'{n},
  Piotr Micek, G\"unter Rote, and Torsten Ueckerdt.
\newblock Coloring hypergraphs induced by dynamic point sets and bottomless
  rectangles.
\newblock In preparation. A preliminary version was presented at EuroCG'12.

\bibitem{BMP}
Peter Brass, William Moser, and J{\'a}nos Pach.
\newblock {\em Research Problems in Discrete Geometry}.
\newblock Springer, 2005.

\bibitem{CPST09}
Xiaomin Chen, J{\'a}nos Pach, Mario Szegedy, and G{\'a}bor Tardos.
\newblock Delaunay graphs of point sets in the plane with respect to
  axis-parallel rectangles.
\newblock {\em Random Struct. Algorithms}, 34(1):11--23, 2009.

\bibitem{F10}
Radoslav Fulek.
\newblock Coloring geometric hypergraph defined by an arrangement of
  half-planes.
\newblock In {\em CCCG}, pages 71--74, 2010.

\bibitem{GV11}
Matt Gibson and Kasturi~R. Varadarajan.
\newblock Optimally decomposing coverings with translates of a convex polygon.
\newblock {\em Discrete {\&} Computational Geometry}, 46(2):313--333, 2011.

\bibitem{K11}
Bal{\'a}zs Keszegh.
\newblock Coloring half-planes and bottomless rectangles.
\newblock {\em Comput. Geom.}, 45(9):495--507, 2012.

\bibitem{KP11}
Bal{\'a}zs Keszegh and D{\"o}m{\"o}t{\"o}r P{\'a}lv{\"o}lgyi.
\newblock Octants are cover-decomposable.
\newblock {\em Discrete {\&} Computational Geometry}, 47(3):598--609, 2012.

\bibitem{KP12}
Bal{\'a}zs Keszegh and D{\"o}m{\"o}t{\"o}r P{\'a}lv{\"o}lgyi.
\newblock Octants are cover-decomposable into many coverings.
\newblock {\em CoRR}, abs/1207.0672, 2012.

\bibitem{Mat}
Jir\'{\i} Matou\v{s}ek.
\newblock {\em Lectures on Discrete Geometry}.
\newblock Springer, 2002.

\bibitem{Pa80}
J{\'a}nos Pach.
\newblock Decomposition of multiple packing and covering.
\newblock In {\em 2. Kolloq. {\"u}ber Diskrete Geom.}, pages 169--178. Inst.
  Math. Univ. Salzburg, 1980.

\bibitem{PT10c}
J{\'a}nos Pach and G{\'a}bor Tardos.
\newblock Tight lower bounds for the size of epsilon-nets.
\newblock {\em CoRR}, abs/1012.1240, 2010.

\bibitem{PTT05}
J{\'a}nos Pach, G{\'a}bor Tardos, and G{\'e}za T{\'o}th.
\newblock Indecomposable coverings.
\newblock In {\em CJCDGCGT}, pages 135--148, 2005.

\bibitem{PT09}
J{\'a}nos Pach and G{\'e}za T{\'o}th.
\newblock Decomposition of multiple coverings into many parts.
\newblock {\em Comput. Geom.}, 42(2):127--133, 2009.

\bibitem{PT10}
D{\"o}m{\"o}t{\"o}r P{\'a}lv{\"o}lgyi and G{\'e}za T{\'o}th.
\newblock Convex polygons are cover-decomposable.
\newblock {\em Discrete {\&} Computational Geometry}, 43(3):483--496, 2010.

\bibitem{SY10}
Shakhar Smorodinsky and Yelena Yuditsky.
\newblock Polychromatic coloring for half-planes.
\newblock In {\em SWAT}, pages 118--126, 2010.

\bibitem{TT07}
G{\'a}bor Tardos and G{\'e}za T{\'o}th.
\newblock Multiple coverings of the plane with triangles.
\newblock {\em Discrete {\&} Computational Geometry}, 38(2):443--450, 2007.

\end{thebibliography}

\end{document}